\documentclass[submission, Phys]{SciPost}
\usepackage{verbatim}
\usepackage{hyperref}
\usepackage{amsmath,braket,mathdots,mathtools,amssymb,xcolor,calligra,color}
\usepackage{amsthm} 
\usepackage{bbm}  
\usepackage{fancyhdr}
\usepackage{float}
\usepackage{bm}
\usepackage{authblk}
\usepackage{color}
\usepackage{authblk}
\usepackage{cite} 
\usepackage{graphicx}

\usepackage{tikz}
\usetikzlibrary{spy}
\usetikzlibrary{patterns}
\usetikzlibrary{calc,arrows,positioning}
\usetikzlibrary{arrows,shadows}
\usetikzlibrary{decorations.pathmorphing}	
\usetikzlibrary{decorations.markings}
\usepackage{pgfplots}

\usepackage{pgfplotstable}
\usepackage{filecontents}
\usepackage[utf8]{inputenc}    

\newcommand{\be}{\begin{equation}}
\newcommand{\ee}{\end{equation}}
\newcommand{\bea}{\begin{eqnarray}}
\newcommand{\eea}{\end{eqnarray}}

\def\XXint#1#2#3{{\setbox0=\hbox{$#1{#2#3}{\int}$}
     \vcenter{\hbox{$#2#3$}}\kern-.5\wd0}}

\let\OLDthebibliography\thebibliography
\renewcommand\thebibliography[1]{
  \OLDthebibliography{#1}
  \setlength{\parskip}{0pt}
  \setlength{\itemsep}{0pt plus 0.3ex}
}

\newcommand{\sign}{\,{\rm sgn}\,}

\newcommand{\cc}{\mathfrak{c}}
\newcommand{\1}{\,\pmb{1}}
\newcommand{\D}[1]{{\rm d}#1}

\newtheorem{theorem}{Theorem}

\newtheorem{property}{Lemma}

\hypersetup{
     colorlinks=true,
     linkcolor=blue,
     filecolor=blue,
     citecolor = green,      
     urlcolor=cyan,
     }

\begin{document}

\begin{center}
{\Large\bf Regularization of energy-dependent pointlike interactions in 1D quantum mechanics
}
\end{center}
\begin{center}
Etienne Granet\textsuperscript{1}
\end{center}
\begin{center}
{\bf 1} Kadanoff Center for Theoretical Physics, University of Chicago, 5640 South Ellis Ave, Chicago, IL 60637, USA
\end{center}
%
\section*{Abstract}
{\bf 
We construct a family of hermitian potentials in 1D quantum mechanics that converges in the zero-range limit to a $\delta$ interaction with an \textit{energy-dependent} coupling. It falls out of the standard four-parameter family of pointlike interactions in 1D. Such classification was made by requiring the pointlike interaction to be hermitian. But we show that although our Hamiltonian is hermitian for the standard inner product when the range of the potential is finite, it becomes hermitian for a \textit{different} inner product in the zero-range limit. This inner product attributes a finite probability (and not probability density) for the particle to be exactly located at the position of the potential. 
Such pointlike interactions can then be used to construct potentials with a finite support with an energy-dependent coupling.


}

\section*{Introduction}
Pointlike interactions have long been part of the theoretical physics landscape, as idealized versions of  short-range potentials that make explicit computations possible, see e.g. \cite{kronig1931quantum,thomas1935the,breit1947the,wu1959ground,lieb1963exact}. Their most physical definition is through \textit{regularization}, namely as the zero-range limit of a family of regular potentials. In 1D the archetypal pointlike interaction is the $\delta$ potential \cite{Fermi1934SopraLS,bethe1935quantum}, to which any potential of the form $\tfrac{1}{a}V(\tfrac{x}{a})$ converges in the zero-range limit $a\to 0$, with coupling $\int V(x)\D{x}$ \cite{10.2307/24714232}. In higher dimensions the definition and regularization of the $\delta$ potential is more involved \cite{berezin1960remark,albeverio1979singular,albeverio1981point}. But in 1D there also exist other pointlike interactions that find applications and that require a more involved functional form and scaling \cite{cheon1999fermion,Yukalov2005FermiBoseMF,granet2022duality}. 

In this context, the question of classifying and regularizing pointlike interactions in 1D naturally arose \cite{Sen2003TheFL,SEBA1986111,kurasov1998finite,Seba1986TheGP,kurasov1996distribution,albeverio1998symmetries,Coutinho1997GeneralizedPI,carreau1999four,Romn1996TheRF,albeverio2000approximation,granet2022regularization}. An appealing approach to this classification is to determine the constraints on the resulting Hamiltonian once the zero-range limit is taken. Given a particle in a potential $V_a(x)$ such that $V_a(x)\to 0$ when $a\to 0$ at fixed $x\neq 0$, the following reasoning is tempting. In the zero-range limit $a\to 0$, the particle does not see any potential for $x\neq 0$, so its Hamiltonian should be merely $\partial_x^2$. The only possible remaining effect of the interaction  is thus to impose some boundary conditions on the wave function at the position of the potential. But since the Hamiltonian is hermitian for $a>0$ these boundary conditions are constrained by hermiticity. Such constructions are called self-adjoint extensions of the operator $\partial_x^2$ \cite{gitman2012self,albeverio1988solvable}, and they have been classified to be given by a four-parameter family of connection conditions on the wave function and its derivatives before and after the position of the potential \cite{Seba1986TheGP,kurasov1996distribution,albeverio1998symmetries}. Besides, the Hamiltonian for $a>0$ is time reversal-symmetric, which restricts this further to a three-parameter family in the zero-range limit \cite{albeverio1998symmetries}. Reciprocally, it has been shown that all this family of connection conditions can indeed be obtained through appropriate regularizations \cite{Cheon:1997rx,Romn1996TheRF}. Hence, according to this reasoning, pointlike interactions in 1D quantum mechanics obtained from a regular potential should exactly be this three-parameter family of connection conditions.  

But there is actually a subtlety in this reasoning. The weak point is that hermiticity is inner product-dependent \cite{SCHOLTZ199274}. Although the Hamiltonian is certainly hermitian for the standard inner product when $a>0$, it can become hermitian \textit{for a different} inner product in the zero-range limit, yielding more general connection conditions than previously argued. In this paper, we precisely construct a family of smooth potentials $V_a(x)$ that are hermitian for the standard inner product, but whose zero-range limit is \textit{not} described by the previous connection conditions. We show that in the limit $a\to 0$ it converges instead to a $\delta$ interaction with an \textit{energy-dependent} coupling. This means that the resulting eigenstates of the Hamiltonian satisfy the connection conditions of a $\delta$ interaction, but with an amplitude that depends on the energy of that eigenstate. Such a dependence on the energy is allowed only because the Hamiltonian becomes hermitian for a different inner product when $a\to 0$. This inner product attributes a finite probability (and not just probability density) for a particle to be located at the position of the potential. Similarly, we present as well a potential $V_a(x)$ that in the zero-range limit converges to the so-called "$\delta'$ interaction" between fermions, but with an energy-dependent coupling. These pointlike interactions can then be used to construct energy-dependent potentials that have a finite support.

We then discuss the dynamics of a particle interacting with this potential and the influence of the non-standard inner product. Because of this inner product, the particle has a tendency to stick to the potential, i.e. can be captured by the potential that acts like a pointlike potential well. When two particles interact with each other via this potential, this means a finite probability of forming a pointlike pair of particles, that is however not stable. We investigate such properties with the two pointlike interactions constructed in this paper.

We  note that quantum mechanics with non-hermitian Hamiltonians or with non-standard inner products has been the object of many studies in the past, see e.g. \cite{SCHOLTZ199274,Mostafazadeh:2001nr,mostafazadeh2010pseudo,ashida2020non}, in particular in $\mathcal{P}\mathcal{T}$-symmetric systems \cite{bender1998real,bender2007making,Mostafazadeh:2003gz} and open systems \cite{rotter2009topical,cohen1968}. But in our case, this unusual inner product results from the limit of Hamiltonians that are perfectly hermitian for the usual inner product. The fact that the pointlike interactions built in this paper fall out of the previously mentioned classification does not result from a choice of an unusual inner product in the first place. Being the zero-range limit of regular potentials, they are on the same footing as the three-parameter family presented above.

We also note that energy-dependent potentials in quantum mechanics have attracted interest for a long time, because relativistic equations can be seen sometimes as stationary Schr\"odinger equations with energy-dependent couplings, see e.g. \cite{RIZOV198559}. Pointlike interactions with an energy dependence have also been introduced before, in the form of higher-derivatives $\delta$ potentials \cite{Griffiths_1993,Coutinho_2005,coutinho2006energy,Lange2014DistributionTF}. However their definition therein is only formal and it is known that the connection conditions obtained by integrating the Schr\"odinger equation with higher derivatives $\delta$ functions are often incorrect, see e.g. \cite{SEBA1986111} and the appendix of \cite{granet2022duality}. To the best of our knowledge, our paper is the first to define energy-dependent potentials constructed from usual, regular and hermitian potentials in non-relativistic quantum mechanics.


\tableofcontents

\section{Definitions and reminders}
\subsection{Pointlike interaction \label{pointlie}}
The wave function $\psi(x,t)$ of a particle in a potential $V(x)$ in a ring $[-L/2,L/2]$ of size $L$ with periodic boundary condition satisfies the Schr\"odinger equation
\begin{equation}\label{schrodinger}
i\partial_t \psi=-\partial^2_x\psi+V(x) \psi\,.
\end{equation}
Let $V_a(x)$ a family of potentials parametrized by $a>0$. We say that it describes a pointlike interaction (or pointlike potential) in the limit $a\to 0$ if 
\begin{equation}
\underset{a\to 0}{\lim}\, \int_{-L/2}^{L/2} f(x) V_a(x)\D{x}=0\,,
\end{equation}
for any integrable function $f$ for which there exists $\epsilon>0$ such that $f(x)=0$ for $x\in[-\epsilon,\epsilon]$ \footnote{Imposing $V_a(x)\to 0$ at fixed $x\neq 0$ when $a\to 0$ is not enough. For example, $V_a(x)=\frac{1}{a}$ for $1<x<1+a$ satisfies this property, but produces a $\delta$ interaction at $x=1$ and not at $x=0$.}. The family $V_a(x)$ is then called a regularization of this pointlike interaction. Given a periodic initial condition $\psi(x,t=0)$, we define the time evolution of the wave function of the particle in this pointlike potential as the limit $a\to 0$ of the solution of \eqref{schrodinger} for $a>0$.

\subsection{Energy functional point of view}
The previous definition of pointlike interactions is a physical definition, since a pointlike potential is just an idealised view of a very short-range potential convenient to the theoretician. However, it does not provide a description of the system once the zero-range limit is taken. In this section, we present a simple and flexible framework for this.

Let us first reformulate the previous problem. We consider $\mathcal{H}$ the Hilbert space of smooth functions on $[-L/2,L/2]$ satisfying periodic boundary conditions, equipped with the standard inner product
\begin{equation}
\langle \phi|\psi\rangle=\int_{-L/2}^{L/2}\phi^*(x)\psi(x)\D{x}\,,
\end{equation}
and define the following energy functional for a real potential $V(x)$
\begin{equation}
\begin{aligned}\label{limit}
H(\psi)&=\int_{-L/2}^{L/2}|\psi'(x)|^2\D{x}+\int_{-L/2}^{L/2} V(x)|\psi(x)|^2\D{x}\,.
\end{aligned}
\end{equation}
We recall that the local minima of $H(\psi)$ at fixed norm $\langle \psi| \psi\rangle=1$ satisfy the stationary Schr\"odinger equation
\begin{equation}\label{schrodingers}
-\partial_x^2\psi(x)+V(x)\psi(x)=E\psi(x)\,,
\end{equation}
with $E$ constrained by the boundary conditions. It is then the value taken by the energy functional on this local minimum $H(\psi)=E$. We can then define a dynamics on the linear space $S$ spanned by the local minima of the energy functional $H$ the following way. If $\psi(x)$ is a local minimum with energy $H(\psi)=E$, then we define its time evolution $\psi(x,t)$ by
\begin{equation}
\psi(x,t)=\psi(x)e^{-iEt}\,,
\end{equation}
and extend this definition by linearity to the full vector space $S$. It follows that in $S$,  the time evolution of wave functions is governed by the Schr\"odinger equation \eqref{schrodinger}.\\ 

\subsubsection*{Example: $\delta$ interaction}
Let us illustrate now how pointlike interactions can be described in this framework. To that end, we consider the textbook example of the $\delta$ interaction. It is a pointlike interaction defined by the regularization $V_a(x)=\frac{c}{a}\1_{|x|<a/2}$ in \eqref{limit}, for $c$ real. The Schr\"odinger equation \eqref{schrodingers} can be straightforwardly solved for $a>0$ since the potential is constant by part, from which one can deduce the form of the wavefunction in the pointlike limit $a\to 0$. It is well-known that the eigenstates in the limit $a\to 0$ satisfy
\begin{equation}\label{cond}
\begin{aligned}
\psi''(x)+E\psi(x)&=0\,,\qquad \text{for }x\neq 0\\
\psi(0+)-\psi(0-)&=0\\
\psi'(0+)-\psi'(0-)&=c\psi(0)\,,
\end{aligned}
\end{equation}
with $E$ constrained by the boundary conditions.\\

Let us show now that these boundary conditions exactly correspond to the local minima of the energy functional
\begin{equation}
H(\psi)=\int_{-L/2}^{L/2}|\psi'(x)|^2\D{x}+c|\psi(0)|^2\,,
\end{equation}
at fixed unit norm $\langle \psi|\psi\rangle =1$. Introducing a Lagrange multiplier $E$, the problem is equivalent to minimizing
\begin{equation}
F(\psi,E)=H(\psi)-E(\langle \psi|\psi\rangle-1)\,,
\end{equation}
without constraining $\psi$. Indeed, minimizing with respect to $E$ imposes $\langle \psi|\psi\rangle=1$, while minimizing with respect to $\psi$ minimizes then $H(\psi)$. We have for a continuous perturbation $\delta\psi$
\begin{equation}
\begin{aligned}
F(\psi+\delta\psi,E)=&F(\psi,E)+2\Re \int_{-L/2}^{L/2}(\psi'(x))^* \delta \psi'(x)\D{x}-2E\Re \int_{-L/2}^{L/2}\psi(x)^* \delta \psi(x)\D{x}\\
&+2c\Re  \psi(0)^* \delta\psi(0)+\mathcal{O}(\delta\psi^2)\,.
\end{aligned}
\end{equation}
Let us assume that there are points $y_j$ where $\psi'$ is discontinuous. Then performing an integration by part
\begin{equation}
\begin{aligned}
F(\psi+\delta\psi,E)&-F(\psi,E)=-2\Re \int_{-L/2}^{L/2}(\psi''(x)+E\psi(x))^* \delta \psi(x)\D{x}\\
&+2c\Re  \psi(0)^* \delta\psi(0)-2\Re \sum_j (\psi'(y_j+)-\psi'(y_j-))^*\delta\psi(y_j)\\
&+\mathcal{O}(\delta\psi^2)\,.
\end{aligned}
\end{equation}
It follows that in order to make vanish all the $\delta\psi$ terms, one requires a continuous $\psi'(x)$ when $x\neq 0$, and
\begin{equation}
\psi'(0+)-\psi'(0-)=c\psi(0)\,,
\end{equation}
as well as for $x\neq 0$
\begin{equation}
\psi''(x)+E\psi(x)=0\,.
\end{equation}

\subsection{Standard pointlike interactions\label{self}}
%
%

We briefly recall the reasoning made in the introduction. The eigenstates of a particle interacting with a pointlike potential at $x=0$ according to the definition of Section \ref{pointlie}, for any $\epsilon>0$ should satisfy a free Schrodinger equation $\psi''(x)+E\psi(x)=0$ for $|x|>\epsilon$. Hence the only possible effect of the pointlike interaction is to impose some boundary conditions for $x>0$ and $x<0$. Since the Hamiltonian is hermitian for $a>0$, it is tempting to look for connection conditions that will make the resulting pointlike interaction hermitian. We will call such settings a standard pointlike interaction at $x=0$. In the math terminology, this corresponds to self-adjoint extensions of the operator $\partial_x^2$ on the space of smooth functions on $[-L/2,L/2]\setminus\{0\}$ \cite{gitman2012self}. It is known that the only such possible connection conditions are \cite{Seba1986TheGP,kurasov1996distribution,albeverio1998symmetries}
\begin{equation}\label{sa}
\begin{cases}
\text{either }\qquad\left(\begin{matrix}\psi(0+)\\ \psi'(0+)\end{matrix}\right)=e^{i\theta}\left(\begin{matrix}a&b\\ c&d\end{matrix}\right)\left(\begin{matrix}\psi(0-)\\ \psi'(0-)\end{matrix}\right)\\
\text{or }\qquad \psi'(0+)=h_+\psi(0+)\,,\qquad \psi'(0-)=h_-\psi(0-)\,,
\end{cases}
\end{equation}
with $\theta,a,b,c,d,h_\pm$ real parameters satisfying $ad-bc=1$. The second possibility in \eqref{sa} can be obtained as particular limits of the first possibility when some parameters go to $\pm\infty$. This leaves thus a four-parameter family. If now time-reversal invariance is imposed, it constrains $\theta=0$, leaving a three-parameter family \cite{albeverio1998symmetries}. Besides, the connection conditions \eqref{sa} are indeed pointlike interactions in the sense of Section \ref{pointlie}, i.e. they can be obtained as limits of regular potentials \cite{Cheon:1997rx}. 

In the following Sections we precisely show that not all pointlike interactions are standard pointlike interactions, i.e. there exist pointlike interactions with more general connection conditions than \eqref{sa}.

\section{Energy-dependent pointlike interaction between bosons \label{running}}

\subsection{The potential}
The first energy-dependent pointlike interaction that we present acts like a $\delta$ potential with an energy-dependent coupling for bosons, and is transparent for fermions. Its construction is given by the following Theorem.

\begin{theorem}\label{thm}
Let $c\geq 0, \nu\geq 0$. Let $\sigma(x)$ a smooth odd increasing function such that $\sigma(x)\to 1$ when $x\to\infty$, and $\Delta(x)$ a smooth positive even function that satisfies
\begin{equation}
\begin{aligned}
&\int_{-\infty}^{\infty}\Delta(x)\D{x}=1\,,\qquad \int_{-\infty}^{\infty}\sqrt{\Delta(x)}\D{x}<\infty\,,\qquad \max(x^2\Delta(x))<\infty\,.
\end{aligned}
\end{equation}
[For example, one can take $\sigma(x)=\tanh(x)$ and $\Delta(x)=\tfrac{1}{2}\tanh'(x)$.] For $E$ real and $a>0$, we define $\psi_a(x)$ as the  solution to
\begin{equation}\label{shrod}
-\psi_a''(x)+V_a(x)\psi_a(x)=E\psi_a(x)\,,
\end{equation}
with some fixed $a$-independent boundary conditions at $x=1$, and where the potential is
\begin{equation}\label{vax}
V_a(x)=\frac{v_a''(x)}{v_a(x)}\,,\qquad v_a(x)=1+\frac{c}{2}x\sigma(\tfrac{x}{a})+\nu\sqrt{\tfrac{1}{a}\Delta(\tfrac{x}{a})}\,.
\end{equation}
The limit $\psi(x)\equiv \underset{a\to 0}{\lim}\, \psi_a(x)$ satisfies
\begin{equation}\label{new}
\begin{aligned}
&\psi''(x)+E\psi(x)=0\quad {\rm for }\,\,x\neq 0\\
&\left(\begin{matrix}\psi(0+)\\ \psi'(0+)\end{matrix}\right)=\left(\begin{matrix}1&0\\ c-E\nu^2&1\end{matrix}\right)\left(\begin{matrix}\psi(0-)\\ \psi'(0-)\end{matrix}\right)\,.
\end{aligned}
\end{equation}
\end{theorem}
Before entering the proof, let us make the following comment on the connection conditions \eqref{new}. If the wave function describes two bosons interacting via $V_a(x)$ in the center of mass frame, then $\psi(x)$ has to be even, so $\psi(0+)=\psi(0-)$ and $\psi'(0+)=-\psi'(0-)$. This leads to the jump condition
\begin{equation}
\begin{aligned}
\psi'(0+)-\psi'(0-)&=(c-E\nu^2)\psi(0)\,,
\end{aligned}
\end{equation}
which is exactly like a $\delta$ potential with an energy-dependent coupling. Now, if the wave function describes fermions, it has to be odd and so $\psi(0+)=-\psi(0-)$ and $\psi'(0+)=\psi'(0-)$. The connection conditions \eqref{new} are then automatically satisfied, meaning that this pointlike potential is transparent for fermions, i.e. has no effect.

Besides, we give a plot of the potential $V_a(x)$ in the right panel of Figure \ref{potential1}.

\begin{proof}[Proof of Theorem \ref{thm}]
Two independent solutions of \eqref{shrod} for $E=0$ are $v_a(x)$ which is even, and
\begin{equation}
w_a(x)=v_a(x)\int_0^x \frac{\D{t}}{v_a(t)^2}\,,
\end{equation}
which is odd. Hence the solutions to
\begin{equation}
-\psi_a''(x)+V_a(x)\psi_a(x)=f(x)\,,
\end{equation}
for $f(x)$ an arbitrary even smooth function, are
\begin{equation}
\psi_a(x)=v_a(x)\int_0^x f(y)w_a(y)\D{y}-w_a(x)\int_0^x f(y)v_a(y)\D{y}+A_av_a(x)+B_aw_a(x)\,.
\end{equation}
with $A_a,B_a$ integration constants depending on the boundary conditions. It follows that $\psi_a(x)$ satisfies the self-consistency equation
\begin{equation}
\begin{aligned}
\psi_a(x)=Ev_a(x)\int_0^x \psi_a(y)w_a(y)\D{y}-Ew_a(x)\int_0^x \psi_a(y)v_a(y)\D{y}+A_av_a(x)+B_aw_a(x)\,.
\end{aligned}
\end{equation}
Let us argue that one can assume that $A_a,B_a$ are independent of $a$. Indeed, by linearity of this integral equation, one can focus on the two cases $A_a=0$ (odd solution) and $B_a=0$ (even solution). Then, focusing e.g. on the case $B_a=0$, if we define $\tilde{\psi}_a$ as the solution to this integral equation with an $A$ fixed independent of $a$, one has $\psi_a(x)=\tfrac{A_a}{A}\tilde{\psi}_a(x)$. Hence this provides a solution to the problem provided $A_a=A\tfrac{\psi_a(1)}{\tilde{\psi}_a(1)}$ has a finite limit when $a\to 0$. 

We will thus assume from now on that $A_a=A$ and $B_a=B$ are independent of $a$. Defining then
\begin{equation}
\bar{\psi}_a(x)=\psi_a(x)-A\nu\sqrt{\tfrac{1}{a}\Delta(\tfrac{x}{a})}\,,
\end{equation}
it satisfies
\begin{equation}\label{rec}
\begin{aligned}
\bar{\psi}_a(x)&=Ev_a(x)\int_0^x \bar{\psi}_a(y)w_a(y)\D{y}-Ew_a(x)\int_0^x \bar{\psi}_a(y)\bar{v}_a(y)\D{y}\\
&-E\nu w_a(x)\int_0^x \bar{\psi}_a(y)\sqrt{\tfrac{1}{a}\Delta(\tfrac{y}{a})}\D{y}\\
&+EA\nu v_a(x)\int_0^x \sqrt{\tfrac{1}{a}\Delta(\tfrac{y}{a})}w_a(y)\D{y}-EA\nu w_a(x)\int_0^x\sqrt{\tfrac{1}{a}\Delta(\tfrac{y}{a})}\bar{v}_a(y)\D{y}\\
&-EA\nu^2 w_a(x)\int_0^x \tfrac{1}{a}\Delta(\tfrac{y}{a})\D{y}+A\bar{v}_a(x)+Bw_a(x)\,,
\end{aligned}
\end{equation}
with
\begin{equation}
\bar{v}_a(x)=1+\frac{c}{2}x\sigma(\tfrac{x}{a})\,.
\end{equation}
To go further, we prove the following Lemma.
\begin{property}\label{limconv}If $|f_a(x)|$ is bounded independently of $a$ and $x\in[-L/2,L/2]$, then
\begin{equation}
\underset{a\to 0}{\lim}\,  \int_{-L/2}^{L/2}f_a(x)\sqrt{\tfrac{1}{a}\Delta(\tfrac{x}{a})}\D{x}=0\,.
\end{equation}
\end{property}
\begin{proof}
Denoting $M$ the bound on $f_a$, we have
\begin{equation}
\left| \int_{-L/2}^{L/2}f_a(x)\sqrt{\tfrac{1}{a}\Delta(\tfrac{x}{a})}\D{x}\right|\leq M  \sqrt{a}\int_{-L/(2a)}^{L/(2a)}\sqrt{\Delta(x)}\D{x}\,,
\end{equation}
which goes to zero when $a\to 0$ from the assumption $\int_{-\infty}^{\infty}\sqrt{\Delta(x)}\D{x}<\infty$.
\end{proof}
This will be useful to prove the following Lemma.
\begin{property}\label{uniform}
The functions $w_a(x)$ and $\bar{\psi}_a(x)$ are bounded independently of $a<1$ and $x\in[-L/2,L/2]$.
\end{property}
\begin{proof}
Let us start with $w_a(x)$, and restrict to $x\geq 0$ by oddness of $w_a(x)$. Since $v_a(x)\geq 1$ by assumptions on $c,\nu,\sigma$, we have
\begin{equation}
0\leq w_a(x)\leq v_a(x) \int_0^x \D{t}=xv_a(x)\,.
\end{equation}
We then write
\begin{equation}
xv_a(x)=x\bar{v}_a(x)+\nu \sqrt{a}\sqrt{\tfrac{x^2}{a^2}\Delta(\tfrac{x}{a})}\,.
\end{equation}
We have $|\bar{v}_a(x)|\leq 1+cL/4$ for all $x\in[-L/2,L/2]$. Since there exists $M>0$ such that $x^2\Delta(x)<M$ for all $x\geq 0$ by assumption, it follows that for $a<1$
\begin{equation}
|w_a(x)|\leq \tfrac{L}{2}(1+cL/4)+\nu \sqrt{M}\,,
\end{equation}
which is a bound independent of $x$ and $a<1$. We note that from this, we are allowed to use uniform convergence theorem in the expression for $w_a(x)$ and take the limit $a\to 0$ under the integral. This yields the limits $a\to 0$ for $x\neq 0$
\begin{equation}
v_a(x)\to 1+\tfrac{c}{2}|x|\,,\qquad w_a(x)\to x\,.
\end{equation}
Let us now focus on $\bar{\psi}_a(x)$ and \eqref{rec}, and consider again $x>0$. Using Lemma \ref{limconv} and the uniform bound on $w_a(x),\bar{v}_a(x)$ we have for all $x\in[-L/2,L/2]$
\begin{equation}
\underset{a\to 0}{\lim}\, Ev_a(x)\int_0^x \sqrt{\tfrac{1}{a}\Delta(\tfrac{y}{a})}w_a(y)\D{y}-E w_a(x)\int_0^x\sqrt{\tfrac{1}{a}\Delta(\tfrac{y}{a})}\bar{v}_a(y)\D{y}=0\,,
\end{equation}
so this term can be bounded independently of $x\in [-L/2,L/2]$ and $a<1$. We have as well
\begin{equation}
\left| -EA\nu^2 w_a(x)\int_0^x \tfrac{1}{a}\Delta(\tfrac{y}{a})\D{y}+A\bar{v}_a(x)+Bw_a(x)\right|\leq \frac{1}{2}EA\nu^2 \max w+|A|\max \bar{v}+|B|\max w\,,
\end{equation}
where $\max w,\bar{v}$ are the uniform bounds of $w_a(x),\bar{v}_a(x)$. It follows then from \eqref{rec} that there exist constants $\lambda,\mu,\eta>0$ independent of $x\in [-L/2,L/2]$ and $a<1$ such that
\begin{equation}\label{eqpsi}
|\bar{\psi}_a(x)|\leq \int_0^x |\bar{\psi_a}(y)|(\lambda+\mu \sqrt{\tfrac{1}{a}\Delta(\tfrac{y}{a})})\D{y}+\eta\,.
\end{equation}
Let us introduce
\begin{equation}
g_a(x)=\lambda+\mu \sqrt{\tfrac{1}{a}\Delta(\tfrac{y}{a})}\,,
\end{equation}
and
\begin{equation}
G_a(x)=e^{\int_0^x g_a(t)\D{t}}\,,
\end{equation}
as well as $h_a(x)$ defined by
\begin{equation}
h_a(x) G_a(x)=\int_0^x |\bar{\psi}_a(y)|g_a(y)\D{y}\,.
\end{equation}
With these notations, \eqref{eqpsi} translates into
\begin{equation}
h_a'(x)\leq \eta\frac{ g_a(x)}{G_a(x)}\,.
\end{equation}
Hence
\begin{equation}
h_a(x)\leq \eta \int_0^x \frac{g_a(t)}{G_a(t)}\D{t}\,.
\end{equation}
Noting that $G_a(x)\geq 1$ we have then
\begin{equation}
0\leq h_a(x)\leq \eta \int_0^x g_a(t)\D{t}\,.
\end{equation}
The right-hand side has a finite limit $\eta\lambda x$ when $a\to 0$ according to Lemma \ref{limconv}, so it can be bounded independently of $a<1$ and $x\in[-L/2,L/2]$.
\end{proof}
Lemma \ref{uniform} allows one to use dominated convergence theorem and swap the integration and limits $a\to 0$ in \eqref{rec}. Denoting for $x\geq 0$, $\bar{\psi}(x)=\underset{a\to 0}{\lim}\, \bar{\psi}_a(x)$, and noting that
\begin{equation}\label{delta}
\underset{a\to 0}{\lim}\, \int_0^x \tfrac{1}{a}\Delta(\tfrac{y}{a})\D{y}=\frac{1}{2}\sign(x)\,,
\end{equation}
we obtain
\begin{equation}
\begin{aligned}
\bar{\psi}(x)&=E(1+\tfrac{c}{2}|x|)\int_0^x \bar{\psi}(y)y\D{y}-Ex\int_0^x \bar{\psi}(y)(1+\tfrac{c}{2}|y|)\D{y}\\
&-\frac{1}{2}EA\nu^2 |x|+A\left(1+\tfrac{c}{2}|x|\right)+Bx\,.
\end{aligned}
\end{equation}
From this one finds $\bar{\psi}(0+)=\bar{\psi}(0-)=A$ and $\bar{\psi}'(0\pm)=B\pm \tfrac{A}{2}(c-E\nu^2)$. Differentiating twice for $x\neq 0$, one finds moreover $\bar{\psi}''(x)+E\bar{\psi}(x)=0$. Hence this is exactly equivalent to
\begin{equation}\label{new2}
\begin{aligned}
\bar{\psi}''(x)+E\bar{\psi}(x)&=0\quad {\rm for }\,\,x\neq 0\\
\bar{\psi}(0+)-\bar{\psi}(0-)&=0\\
\bar{\psi}'(0+)-\bar{\psi}'(0-)&=(c-E\nu^2)\bar{\psi}(0)\,.
\end{aligned}
\end{equation}
Since for all $x\neq 0$ we have $\psi(x)=\bar{\psi}(x)$, this concludes the proof of the Theorem.

\end{proof}


\subsection{The resulting inner product}
Because of the energy dependence in the amplitude of the cusp in  \eqref{new}, this cannot be put into the form of the standard pointlike interactions whose connection conditions are classified by \eqref{sa}. This could seem contradictory, since \eqref{new} are obtained as the zero-range limit of hermitian Hamiltonians and \eqref{sa} are precisely obtained by imposing hermiticity. To resolve this contradiction, let us consider a particle in a ring $[-L/2,L/2]$ of size $L$ with periodic boundary conditions, interacting with this pointlike potential at position $0$. One finds that the even eigenstates are
\begin{equation}\label{eig}
\psi^{[k]}(x)=\frac{1}{Z_k}\left(\cos(kx)+\alpha(k)\sin(k|x|)\right)\,,
\end{equation}
with
\begin{equation}\label{alpha}
\alpha(k)=\frac{c}{2k}-\frac{k\nu^2}{2}\,,
\end{equation}
with $Z_k$ a normalization factor. They have energy $E_k=k^2$, and periodic boundary conditions impose the following quantization condition on $k$
\begin{equation}\label{quant}
\tan(kL/2)=\alpha(k)\,.
\end{equation}
As for the odd eigenstates, they are
\begin{equation}
\phi^{[k]}(x)=\frac{\sin(kx)}{Z'_k}\,,
\end{equation}
with $Z'_k$ a normalization factor, and where the quantization condition is $\sin(kL/2)=0$. Crucially, one checks that the eigenstates \eqref{eig} are \textit{not} orthogonal with respect to the standard inner product
\begin{equation}
\langle \psi^{[k]}|\psi^{[q]}\rangle\neq 0\quad \text{if } k\neq q\,.
\end{equation}
Yet, for $a>0$, denoting $\psi^{[k]}_a(x)$ the corresponding eigenstate, we certainly have
\begin{equation}
\langle \psi^{[k]}_a|\psi^{[q]}_a\rangle= 0\quad \text{if } k\neq q\,,
\end{equation}
since the $\psi^{[k]}_a$'s are eigenstates of a hermitian Hamiltonian with respect to the standard inner product. In fact, writing $\psi^{[k]}_a=\bar{\psi}^{[k]}_a+A^{[k]}\nu\sqrt{\tfrac{1}{a}\Delta(\tfrac{1}{a})}$ and noting that $\bar{\psi}^{[k]}_a$ converges uniformly from Lemma \ref{uniform}, we have
\begin{equation}
\int_{-L/2}^{L/2}\bar{\psi}^{[k]*}(x)\bar{\psi}^{[q]}(x)\D{x}+A^{[k]*}A^{[q]}\nu^2=0\,.
\end{equation}
Since $A^{[k]}=\psi^{[k]}(0)$, it follows that the eigenstates $\psi^{[k]}$ are orthogonal with respect to the inner product
\begin{equation}\label{inner}
[ \phi|\psi]=\int_{-L/2}^{L/2}\phi^*(x)\psi(x)\D{x}+\nu^2\phi^*(0)\psi(0)\,.
\end{equation}
%
%
Namely, with the normalization factor $Z_k$
\begin{equation}\label{z}
Z_k^2=\frac{L}{2}(1+\alpha(k)^2)+\frac{\alpha(k)}{k}+\nu^2\,,
\end{equation}
we have
\begin{equation}
[\psi^{[k]}|\psi^{[q]}]=\delta_{kq}\,.
\end{equation}
The physical meaning of this inner product is that in the zero-range limit $a\to 0$, the particle has a finite \textit{probability} (and not probability density) to be located at the position of the potential $x=0$, given by
\begin{equation}
\nu^2 |\psi(0)|^2\,.
\end{equation}
The pointlike interaction can thus be thought as an idealized version of a potential well that has a finite probability of capturing a particle.


\subsection{The resulting Hamiltonian \label{hamiy}}
Let us show that the local minima of the energy functional 
\begin{equation}\label{eqham}
H(\psi)=\int_{-L/2}^{L/2}|\psi'(x)|^2\D{x}+c|\psi(0)|^2\,,
\end{equation}
at fixed unit norm $[\psi|\psi]=1$ where this non-standard inner product is defined in \eqref{inner}, satisfy the connection conditions \eqref{new}. Introducing a Lagrange multiplier $E$, we minimize
\begin{equation}
F(\psi,E)=H(\psi)-E([ \psi|\psi]-1)\,.
\end{equation}
We have up to $(\delta\psi)^2$ corrections
\begin{equation}
\begin{aligned}
F(\psi+\delta\psi,E)=&F(\psi,E)-\int_{-L/2}^{L/2}(\psi''(x)+E\psi(x))^*\delta\psi(x)\D{x}\\
&+2\Re[(c\psi(0)-\psi'(0+)+\psi'(0-))^*\delta\psi(0)]\\
&-2E\nu^2\Re[\psi(0)^* \delta \psi(0)]\,.
\end{aligned}
\end{equation}
For this to vanish for all $\delta\psi$, one requires the conditions \eqref{new} indeed.\\


\subsection{Two particles interacting via the pointlike potential \label{2}}
The problem of two particles interacting via a potential straightforwardly reduces to a one-particle problem in the center of mass frame. However, in our case the zero-range limit brings some additional peculiarities that need to be discussed. We thus consider two particles that interact via $V_a(x)$ in \eqref{vax}, i.e. with Hamiltonian \footnote{We omitted the factor $1/2$ in \eqref{eqham} compared to \eqref{hamu} and in the definition in Section \ref{pointlie} to anticipate the two-particle case with a factor $1/2$ which in the center of mass frame reduces to \eqref{schrodinger} without factor $1/2$.}
\begin{equation}\label{hamu}
\begin{aligned}
H_a=&\frac{1}{2}\iint_{-L/2}^{L/2} \left(|\partial_{x_1} \psi(x_1,x_2)|^2+|\partial_{x_2} \psi(x_1,x_2)|^2\right)\D{x_1}\D{x_2}\\
&+\iint_{-L/2}^{L/2} V_a(x_1-x_2)|\psi(x_1,x_2)|^2\D{x_1}\D{x_2}\,.
\end{aligned}
\end{equation}
The local minima $\psi_a(x_1,x_2)$ of this Hamiltonian satisfy the Schr\"odinger equation
\begin{equation}
-\frac{1}{2}(\partial_{x_1}^2+\partial_{x_2}^2) \psi_a+V_a(x_1-x_2)\psi_a=E\psi_a\,,
\end{equation}
with $E$ to determine from the boundary conditions. Writing
\begin{equation}
\psi_a(x_1,x_2)=\phi_a(x)e^{iky}\,,
\end{equation}
with
\begin{equation}
x=x_2-x_1\,,\qquad y=x_1+x_2\,,
\end{equation}
we have that $\phi_a(x)$ satisfies
\begin{equation}
-\phi_a''(x)+V_a(x)\phi_a(x)=(E-k^2)\phi_a(x)\,.
\end{equation}
Hence Theorem \ref{thm} applies and in the limit $a\to 0$, $\phi(x)\equiv \underset{a\to 0}{\lim} \phi_a(x)$ satisfies
\begin{equation}
\phi'(0+)-\phi'(0-)=(c-\nu^2(E-k^2))\phi(0)\,.
\end{equation}
Equivalently, this jump condition can be written
\begin{equation}
\phi'(0+)-\phi'(0-)=c\phi(0)+\nu^2\phi''(0)\,,
\end{equation}
which is, in terms of the original function
\begin{equation}\label{boundbound}
\left.\left[\partial_{x_2}-\partial_{x_1}-c-\frac{\nu^2}{4}(\partial_{x_2}-\partial_{x_1})^2\right] \psi(x_1,x_2)\right|_{x_2=x_1+}=0\,.
\end{equation}
In this limit, the eigenstates of the Hamiltonian are orthogonal with respect to the inner product
\begin{equation}
[ \phi|\psi]=\iint_{-L/2}^{L/2}\phi^*(x_1,x_2)\psi(x_1,x_2)\D{x_1}\D{x_2}+\nu^2\int_{-L/2}^{L/2}\phi^*(x,x)\psi(x,x)\D{x}\,.
\end{equation}
Physically, the additional piece proportional to $\nu^2$ in this inner product corresponds to a finite probability that the two particles form a pointlike pair of particles. However this bound state is not stable as we will see in Section \ref{timeevolution}.

The expression of the Hamiltonian brings however an important difference compared to the one-particle case, as is not the same as for a simple $\delta$ interaction \eqref{eqham}. Let us show that the local minima of the following Hamiltonian
\begin{equation}
\begin{aligned}
H(\psi)&=\frac{1}{2}\iint_{-L/2}^{L/2} \left( |\partial_{x_1} \psi(x_1,x_2)|^2+|\partial_{x_2} \psi(x_1,x_2)|^2\right)\D{x_1}\D{x_2}\\
&+\frac{\nu^2}{4} \int_{-L/2}^{L/2} |(\partial_{x_1}+\partial_{x_2})\psi(x,x)|^2\D{x}\\
&+c\int_{-L/2}^{L/2}|\psi(x,x)|^2\D{x}
\end{aligned}
\end{equation}
at fixed unit norm $[\psi|\psi]=1$ implement the boundary conditions \eqref{boundbound}. To that end, we introduce a Lagrange multiplier $E$ and minimize
\begin{equation}
F(\psi,E)=H(\psi)-E([\psi|\psi]-1)\,.
\end{equation}
We have up to $\mathcal{O}(\delta\psi^2)$ corrections, with simplified notations $\partial_j=\partial_{x_j}$ and where the integrations are on $x_1,x_2$ in $[-L/2,L/2]\times [-L/2,L/2]$ 
\begin{equation}
\begin{aligned}
F(\psi+\delta \psi,E)-F(\psi,E)=&-2\Re \int\left[E\psi^* \delta\psi-\frac{1}{2} \partial_1 \psi^* \partial_1 \delta\psi- \frac{1}{2} \partial_2\psi^* \partial_2 \delta\psi\right]\\
&+2\Re \int\delta(x_1-x_2)\left[ (c-E\nu^2) \psi^*\delta\psi+\frac{\nu^2}{4} (\partial_1+\partial_2)\psi^* (\partial_1+\partial_2) \delta\psi  \right]\,.
\end{aligned}
\end{equation}
Performing integration by parts to remove all derivatives of $\delta\psi$, we have
\begin{equation}
\begin{aligned}
F(\psi+\delta \psi,E)-F(\psi,E)=&-2\Re \int\left[E\psi^* +\frac{1}{2} \partial^2_1 \psi^*+\frac{1}{2} \partial^2_2 \psi^*\right]\delta\psi\\
&+\Re \int  \partial_1 \psi^* \delta\psi [\delta(x_1-(x_{2}+))-\delta(x_1-(x_{2}-))]\\
&+\Re \int \partial_2 \psi^* \delta\psi [\delta(x_2-(x_{1}+))-\delta(x_2-(x_{1}-))]\\
&+2\Re \int\delta(x_1-x_2)\left[ (c-E\nu^2) \psi^*-\frac{\nu^2}{4} (\partial_1+\partial_2)^2\psi^*   \right]\delta\psi\,.
\end{aligned}
\end{equation}
Using that $\psi$ is symmetric in its arguments, this is
\begin{equation}
\begin{aligned}
F(\psi+\delta \psi,E)&-F(\psi,E)=-2\Re \int\left[E\psi^* +\frac{1}{2} \partial^2_1 \psi^*+\frac{1}{2} \partial^2_2 \psi^*\right]\delta\psi\\
&+2\Re \int\delta(x_1-x_2)\left[ (\partial_1-\partial_2)\psi^*+(c-E\nu^2) \psi^*-\frac{\nu^2}{4} (\partial_1+\partial_2)^2\psi^*   \right]\delta\psi\,.
\end{aligned}
\end{equation}
Hence imposing the cancellation of the $\delta\psi$ terms yields
\begin{equation}
\begin{aligned}
&-\frac{1}{2}(\partial^2_1 +\partial^2_2)\psi=E\psi\\
&\left.\left[\partial_{2}-\partial_{1}-c+E\nu^2+\frac{\nu^2}{4}(\partial_{2}+\partial_{1})^2\right] \psi(x_1,x_2)\right|_{x_2=x_1+}=0\,.
\end{aligned}
\end{equation}
Using the first line, the second line is exactly the connection condition \eqref{boundbound}.\\

We see that there is an additional kinetic term proportional to $\nu^2$ in the Hamiltonian. Physically, this term corresponds to the kinetic energy of the pointlike pair of particles that can be formed.

\subsection{Dynamics \label{timeevolution}}
In order to illustrate some properties of the pointlike interaction constructed in Theorem \ref{thm}, we consider the simple problem of the dynamics of a particle in a ring $[-L/2,L/2]$ interacting with the pointlike potential at $x=0$, or equivalently two particles interacting with each other via the pointlike potential.

Given an initial condition $\psi^{\rm ini}(x)$, the usual way to compute its time evolution is to decompose it along the eigenstates of the Hamiltonian. However in our case the unusual inner product requires some precautions. The safe way to understand the decomposition with this unusual inner product is to work at fixed $a>0$ and take the limit $a\to 0$. Let us thus consider an $a$-dependent initial condition $\psi^{\rm ini}_a(x)$, that we will assume to be even for simplicity of the discussion, and that we decompose along the eigenstates at $a>0$
\begin{equation}
\psi^{\rm ini}_a(x)=\sum_k \psi^{[k]}_a(x) \langle \psi^{[k]}_a| \psi^{\rm ini}_a\rangle\,.
\end{equation}
The question is then the limit $a\to 0$ of the overlap $\langle \psi^{[k]}_a| \psi^{\rm ini}_a\rangle$. As seen from Lemma \ref{uniform}, $\psi^{[k]}_a(x)$ contains a piece $\psi^{[k]}(0) \nu \sqrt{\tfrac{1}{a}\Delta(\tfrac{x}{a})}$ that does not converge uniformly. If $\psi^{\rm ini}_a$ itself converges uniformly to $\psi^{\rm ini}$, then from Lemma \ref{limconv} we have
\begin{equation}
\underset{a\to 0}{\lim}\, \langle \psi^{[k]}_a| \psi^{\rm ini}_a\rangle=\langle \psi^{[k]}| \psi^{\rm ini}\rangle\,.
\end{equation}
This leads to the time evolution in the $a\to 0$ limit
\begin{equation}
\psi(x,t)=\sum_k \psi^{[k]}(x) \langle \psi^{[k]}| \psi^{\rm ini}\rangle e^{-itE_k}\,,
\end{equation}
with an inner product that is \textit{not} the one for which the eigenstates $\psi^{[k]}$ are orthogonal. However, if the initial condition $\psi^{\rm ini}_a$ contains itself a piece of the form $\psi^{\rm ini}(0)\nu \sqrt{\tfrac{1}{a}\Delta(\tfrac{x}{a})}$, then
\begin{equation}
\underset{a\to 0}{\lim}\, \langle \psi^{[k]}_a| \psi^{\rm ini}_a\rangle=[\psi^{[k]}| \psi^{\rm ini}]\,,
\end{equation}
and the time evolution is written in terms of the modified inner product
\begin{equation}\label{psipsi}
\psi(x,t)=\sum_k \psi^{[k]}(x) [ \psi^{[k]}| \psi^{\rm ini}] e^{-itE_k}\,.
\end{equation}

Let us illustrate the second situation with the following problem. We would like to consider the case in which the particle is located at $x=0$ with probability $1$ at time $t=0$. This can be achieved with  $\psi^{\rm ini}_a(x)=\sqrt{\tfrac{1}{a}\Delta(\tfrac{x}{a})}$.
%
One finds then
\begin{equation}
\underset{a\to 0}{\lim}\, \langle \psi^{[k]}_a| \psi^{\rm ini}_a\rangle=\frac{\nu}{Z_k}\equiv [\psi^{[k]}|\psi^{\rm ini}]\,,
\end{equation}
where $Z_k$ is given in \eqref{z}. For finite $L$, the non-trivial quantization condition \eqref{quant} has to be solved numerically and \eqref{psipsi} evaluated numerically. This way one can check for example the completeness relation $\sum_k |[\psi^{\rm ini}|\psi^{[k]}]|^2=1$ that holds for the non-standard inner product $[.|.]$. As an illustration, we plot in the right panel of Figure \ref{potential1} the probability $p(t)$ that the particle is exactly located at $x=0$ (and not just probability density) given by
\begin{equation}
p(t)=\nu^2 |\psi(0,t)|^2\,.
\end{equation}
%
%
%
The case $L\to\infty$ deserves attention since it is tractable. In this limit the quantization condition becomes
\begin{equation}
k=\frac{2\pi (n+1/2)}{L}+\mathcal{O}(L^{-1})\,,
\end{equation}
with $n\geq 0$ an integer. The time evolution of the wave function can then be analytically computed in a straightforward way and is given by
\begin{equation}
\psi(x,t)=\frac{\nu}{\pi} \int_0^{\infty}\frac{e^{-itk^2}}{1+\alpha(k)^2}\left(\cos(kx)+\alpha(k)\sin(k|x|) \right)\D{k}\,,
\end{equation}
where $\alpha(k)$ is defined in \eqref{alpha}. In particular, the probability $p(t)$ of finding the particle exactly at $x=0$ is 
\begin{equation}
p(t)=\frac{\nu^4}{\pi^2}\left|  \int_0^{\infty}\frac{e^{-itk^2}}{1+\alpha(k)^2}\D{k}\right|^2\,.
\end{equation}
We see that it goes to $0$ when $t\to\infty$ as $t^{-3}$. This means that the pointlike pair of particles formed by two particles interacting via this pointlike potential is unstable and eventually decays with probability $1$.

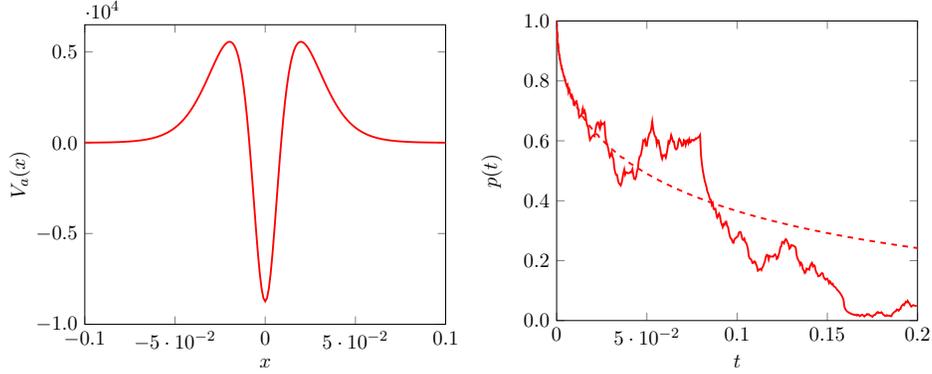
\begin{figure}[h!]
\begin{center}
\begin{tikzpicture}[scale=0.7]
\begin{axis}[
    enlargelimits=false,
    xlabel = $x$,
    ylabel=$V_a(x)$,
    xmax=0.1,
    ymax=6500,
    ymin=-10000,
     y tick label style={
        /pgf/number format/.cd,
            fixed,
            fixed zerofill,
            precision=1,
        /tikz/.cd
    }
]
\addplot[
    line width=1pt,
    color=red,
    opacity=1]
table[
           x expr=\thisrowno{0}, 
           y expr=\thisrowno{1},
         ]{potential.dat};
\end{axis}
\end{tikzpicture}
\begin{tikzpicture}[scale=0.7]
\begin{axis}[
    enlargelimits=false,
    xlabel = $t$,
    ylabel=$p(t)$,
    xmax=0.2,
    ymax=1,
    ymin=0,
     y tick label style={
        /pgf/number format/.cd,
            fixed,
            fixed zerofill,
            precision=1,
        /tikz/.cd
    }
]
\addplot[
    line width=1pt,
    color=red,
    opacity=1]
table[
           x expr=\thisrowno{0}, 
           y expr=\thisrowno{1},
         ]{probability_L0.dat};
\addplot[
    line width=1pt,
    color=red,
    dashed,
    opacity=1]
table[
           x expr=\thisrowno{0}, 
           y expr=\thisrowno{1},
         ]{probability_Linf.dat};
\end{axis}
\end{tikzpicture}
\caption{Left: plot of the regulated potential \eqref{vax} with $c=1$, $\nu=1$, $a=0.01$ and with the functions $\sigma(x)=\tanh(x)$, $\Delta(x)=\tfrac{1}{2}\tanh'(x)$. Right: probability of finding the particle at $x=0$ as a function of time, with $c=1$, $\nu=1$, for a system size $L=1$ (plain) and $L=\infty$ (dashed) [see Section \ref{timeevolution}].} 
\label {potential1}
\end{center}
\end {figure}


\subsection{Energy-dependent potentials with a finite support}
Let us now consider a particle interacting with $N$ pointlike interactions of the type of Theorem \ref{thm} centered at positions $x_n$ with coupling strength $c_n$ and energy-dependence $\nu_n$, for $n=1,...,N$. The stationary states $\psi_N$ are continuous functions that satisfy
\begin{equation}\label{freesh}
\psi_N''(x)+E\psi_N(x)=0\,,\qquad \text{for }x\notin \{x_1,...,x_N\}\,,
\end{equation}
and
\begin{equation}
\psi'_N(x_n+)-\psi_N'(x_n-)=(c_n-E\nu_n^2)\psi_N(x_n)\,.
\end{equation}
We consider now $x_n= -L/2+(n-1/2)\frac{L}{N}$, couplings $c_n=\frac{L}{N}\mathcal{V}(x_n)$ where $\mathcal{V}(x)\geq 0$ is a fixed arbitrary potential, and $\nu^2_n=\frac{L}{N}\nu^2\mathcal{V}(x_n)$ with some $\nu\geq 0$. Between two points $x_n,x_{n+1}$, $\psi_N$ satisfies \eqref{freesh} so that
\begin{equation}
\psi'_N(x_{n+1}-)=\psi_N'(x_n+)-E\psi_N(x_n) \frac{L}{N}+\mathcal{O}(N^{-2})\,.
\end{equation}
Hence
\begin{equation}
\frac{\psi'_N(x_{n+1}-)-\psi'_N(x_n-)}{L/N}=-E\psi_N(x_n)+\mathcal{V}(x_n)(1-E\nu^2)\psi_N(x_n)+\mathcal{O}(N^{-1})\,.
\end{equation}
It follows that in the limit $N\to\infty$, $\psi(x)\equiv \psi_{\infty}(x)$ satisfies
\begin{equation}
-\psi''(x)+\mathcal{V}(x)(1-E\nu^2)\psi(x)=E\psi(x)\,,
\end{equation}
which is a Schr\"odinger equation in a (non pointlike) potential $\mathcal{V}(x)$ with an energy-dependent coupling. This construction, together with Theorem \ref{thm}, provides a regularization of energy-dependent potentials in terms of hermitian, regular potentials in non-relativistic 1D quantum mechanics.

\section{Energy-dependent pointlike interaction between fermions \label{running2}}
\subsection{The potential}
The potential presented in Section \ref{running} is transparent for fermions. We recall that the only non-trivial standard pointlike interaction between fermions induces a discontinuity in the wave function $\psi$ itself, namely
\begin{equation}\label{ferm}
\psi(0+)-\psi(0-)=2\beta \psi'(0)\,,
\end{equation}
with $\beta$ some coupling constant. The continuity of the derivative at $0$ is ensured by the oddness of the wave function for fermions. We will call $\eta$ this pointlike interaction. Depending on the regularization, this $\eta$ potential can act like an infinite wall for bosons \cite{granet2022duality,granet2022regularization} or be transparent \cite{Cheon:1997rx}.\\

The second energy-dependent pointlike interaction that we present in this paper acts like an infinite wall for bosons, and like a $\eta$ potential with an energy-dependent coupling for fermions. Its construction is given by the following Theorem.

\begin{theorem}\label{thm2}
We consider $c,\nu,\sigma,\Delta$ as defined in Theorem \ref{thm}, with moreover $\sigma'(0)>0$. For $E$ real and $a>0$, we define $\psi_a(x)$ as the solution to
\begin{equation}\label{shrod22}
-\psi_a''(x)+V_a(x)\psi_a(x)=E\psi_a(x)\,,
\end{equation}
with some fixed $a$-independent boundary conditions at $x=1$, and where the potential is
\begin{equation}\label{vax2}
V_a(x)=\frac{v_a''(x)}{v_a(x)}\,,\qquad v_a(x)=x+\frac{2}{c}\sigma(\tfrac{x}{a})+\nu\sigma(\tfrac{x}{a})\sqrt{\tfrac{1}{a}\Delta(\tfrac{x}{a})}\,.
\end{equation}
The limit $\psi(x)\equiv \underset{a\to 0}{\lim}\, \psi_a(x)$ satisfies
\begin{equation}\label{newf}
\begin{aligned}
&\psi''(x)+E\psi(x)=0\quad {\rm for }\,\,x\neq 0\\
&\left(\begin{matrix}\psi(0+)\\ \psi'(0+)\end{matrix}\right)=-\left(\begin{matrix}1&0\\ c-E\nu^2&1\end{matrix}\right)\left(\begin{matrix}\psi(0-)\\ \psi'(0-)\end{matrix}\right)
\end{aligned}
\end{equation}
\end{theorem}
Before entering the proof, let us make the following comment on the connection conditions \eqref{newf}.  If the wave function describes two fermions interacting via $V_a(x)$ in the center of mass frame, then $\psi(x)$ has to be odd, so $\psi(0+)=-\psi(0-)$ and $\psi'(0+)=\psi'(0-)$. This leads to the jump condition
\begin{equation}\label{newf2}
\begin{aligned}
\psi(0+)-\psi(0-)&=\frac{4}{c-E\nu^2}\psi'(0)\,,
\end{aligned}
\end{equation}
which is exactly \eqref{ferm} with an energy-dependent coupling. For bosons, $\psi$ has to be even, so $\psi(0+)=\psi(0-)$ and $\psi'(0+)=-\psi'(0-)$. The connection condition \eqref{newf} then implies $\psi(0+)=\psi(0-)=0$, corresponding to hard core bosons. Hence the pointlike potential acts like an infinite wall for bosons.

We also note that under the Girardeau mapping $\psi(x)=\sign(x)\phi(x)$ \cite{girardeau1960relationship,cheon1999fermion,Yukalov2005FermiBoseMF}, if a wave function $\psi$ satisfies \eqref{new}, then $\phi$ satisfies \eqref{newf}. This shows that the pointlike interaction of Theorem \ref{thm2} is dual to that of Theorem \ref{thm} under the Bose-Fermi duality that is the Girardeau mapping in 1D.

\begin{proof}[Proof of Theorem \ref{thm2}]
The proof goes along the same lines as for Theorem \ref{thm}, so that we will only give the main steps.

One first determines two independent solutions to \eqref{shrod22} when $E=0$, that are $v_a(x)$ that is odd, and
\begin{equation}
w_a(x)=\frac{1}{v_a'(x)}+v_a(x)\int_{0}^x \frac{v_a''(y)}{v_a(y)v_a'(y)^2}\D{y}\,,
\end{equation}
that is even. A study of $w_a(x)$ using the assumption $\sigma'(0)>0$ shows that it can be bounded independently of $a<1$ and $x\in [-L/2,L/2]$, and yields the limits when $a\to 0$
\begin{equation}
v_a(x)\to x+\tfrac{c}{2}\sign(x)\,,\qquad w_a(x)\to -\frac{2}{c}|x|\,.
\end{equation}
Then using the same arguments, $\psi_a$ satisfies the following self-consistency equation
\begin{equation}
\begin{aligned}
\psi_a(x)=Ev_a(x)\int_0^x \psi_a(y)w_a(y)\D{y}-Ew_a(x)\int_0^x \psi_a(y)v_a(y)\D{y}+Av_a(x)+Bw_a(x)\,.
\end{aligned}
\end{equation}
with $A,B$ integration constants that can be assumed to be independent of $a$ up to a global rescaling. Then one defines
\begin{equation}
\bar{\psi}_a(x)=\psi_a(x)-A\nu\sigma(\tfrac{x}{a})\sqrt{\tfrac{1}{a}\Delta(\tfrac{x}{a})}\,,
\end{equation}
and shows similarly to Lemma \ref{uniform} that it can be bounded independently of $a<1$ and $x\in[-L/2,L/2]$. This allows one to take the limit $a\to 0$ in the self-consistency equation when rewritten in terms of $\bar{\psi}_a$. Using then again \eqref{delta}, one shows straightforwardly that the resulting self-consistency equation on $\bar{\psi}=\underset{a\to 0}{\lim} \bar{\psi}_a$ is equivalent to the connection conditions \eqref{newf}.

\end{proof}
\color{black}

\subsection{The resulting inner product and Hamiltonian}
In a ring $[-L/2,L/2]$ with periodic boundary conditions, the even eigenstates of the Hamiltonian with a pointlike interaction at $x=0$ are
\begin{equation}
\psi^{[k]}(x)=\frac{\sin (k |x|)}{Z_k'}\,,
\end{equation}
with $Z'_k$ a normalization factor, and where the quantization condition is $\cos(kL/2)=0$. The odd eigenstates are
\begin{equation}
\phi^{[k]}(x)=\frac{1}{Z_k}(\sign(x)\cos(kx)+\alpha(k) \sin (k x))\,,
\end{equation}
with $Z_k$ a normalization factor and $\alpha(k)$ as in \eqref{alpha}, and with
\begin{equation}\label{quant2}
\tan(kL/2)=-\frac{1}{\alpha(k)}\,.
\end{equation}
The eigenstates are orthogonal with respect to the inner product
\begin{equation}\label{inner2}
[ \phi|\psi]=\int_{-L/2}^{L/2}\phi^*(x)\psi(x)\D{x}+\frac{\nu^2}{4}(\phi^*(0+)-\phi^*(0-))(\psi(0+)-\psi(0-))\,.
\end{equation}
As for the resulting Hamiltonian, one obtains with a similar calculation as in Section \ref{hamiy} that the eigenstates are produced as local minima of the energy functional
\begin{equation}
H(\psi)=\int_{-L/2}^{L/2}|\psi'(x)|^2\D{x}+\frac{c}{4}|\psi(0+)-\psi(0-)|^2\,,
\end{equation}
at fixed unit norm $[ \psi|\psi] =1$.


\section*{Summary and perspectives}
In this paper we constructed a potential $V_a(x)$ with range $a$ that converges in the zero-range limit $a\to 0$ to the pointlike Dirac $\delta$ potential with an energy-dependent coupling. We constructed as well a potential that converges to the fermionic "$\delta'$" potential with an energy-dependent coupling. These results are noteworthy for several reasons. 

Firstly, it reveals a new class of pointlike interactions in 1D quantum mechanics, that were previously thought to be classified into a four-parameter family. The inconsistency is resolved as follows. This classification relied on the idea that the Hamiltonian for a pointlike potential should be hermitian for the standard inner product. But we show that our Hamiltonian, despite being indeed hermitian for the standard inner product when $a>0$, becomes hermitian with respect to a different inner product in the limit $a\to 0$, which allows for the pointlike interaction studied in this paper. This raises thus the question of classifying all the possible pointlike interactions in 1D that can be obtained through the zero-range limit of regular potentials. For example, one could ask whether it is possible to obtain a $\delta$ interaction with a more general $E$ dependence in the coupling than what is constructed in this paper. One could also ask whether such construction can be performed in higher dimensions.

Secondly, this pointlike interaction can be interpreted as a pointlike well that has a finite probability of capturing a particle. Indeed, the resulting inner product for which it is hermitian attributes a finite probability (and not just probability density) for the particle to be located at the position of the potential. This pointlike interaction can thus model an inter-particle force that allows for the formation of a pair of particles with a finite probability (although such pairs are unstable). This model could find applications to describe the effect of formation of diatomic molecules in actual Bose gases \cite{bouchoule2020the}. A natural question in that direction is whether a Hamiltonian for $N$ particles interacting via this pointlike interaction enjoy an inner product with a similar physical interpretation. 


Thirdly, this potential shows that running couplings, i.e. energy-dependent couplings are not a specificity of quantum field theories or relativistic quantum mechanics, but also exist in 1D non-relativistic quantum mechanics. To the best of our knowledge this observation has not been made before. It is certainly interesting to investigate whether relations between these pointlike interactions and quantum field theories can be drawn beyond this observation.

\bibliography{bibliographie}

\begin{thebibliography}{10}
\providecommand{\url}[1]{\texttt{#1}}
\providecommand{\urlprefix}{URL }
\expandafter\ifx\csname urlstyle\endcsname\relax
  \providecommand{\doi}[1]{doi:\discretionary{}{}{}#1}\else
  \providecommand{\doi}{doi:\discretionary{}{}{}\begingroup
  \urlstyle{rm}\Url}\fi
\providecommand{\eprint}[2][]{\url{#2}}

\bibitem{kronig1931quantum}
R.~de~L.~Kronig and W.~G. Penney,
\newblock \emph{Quantum mechanics of electrons in crystal lattices},
\newblock Proceedings of the Royal Society of London. Series A, Containing
  Papers of a Mathematical and Physical Character \textbf{130}(814), 499
  (1931).

\bibitem{thomas1935the}
L.~H. Thomas,
\newblock \emph{The interaction between a neutron and a proton and the
  structure of ${\mathrm{h}}^{3}$},
\newblock Phys. Rev. \textbf{47}, 903 (1935),
\newblock \doi{10.1103/PhysRev.47.903}.

\bibitem{breit1947the}
G.~Breit,
\newblock \emph{The scattering of slow neutrons by bound protons. i. methods of
  calculation},
\newblock Phys. Rev. \textbf{71}, 215 (1947),
\newblock \doi{10.1103/PhysRev.71.215}.

\bibitem{wu1959ground}
T.~T. Wu,
\newblock \emph{Ground state of a bose system of hard spheres},
\newblock Phys. Rev. \textbf{115}, 1390 (1959),
\newblock \doi{10.1103/PhysRev.115.1390}.

\bibitem{lieb1963exact}
E.~H. Lieb and W.~Liniger,
\newblock \emph{Exact analysis of an interacting bose gas. i. the general
  solution and the ground state},
\newblock Phys. Rev. \textbf{130}, 1605 (1963),
\newblock \doi{10.1103/PhysRev.130.1605}.

\bibitem{Fermi1934SopraLS}
E.~Fermi,
\newblock \emph{Sopra lo spostamento per pressione delle righe elevate delle
  serie spettrali},
\newblock Il Nuovo Cimento (1924-1942) \textbf{11}, 157 (1934).

\bibitem{bethe1935quantum}
R.~P. H.~Bethe,
\newblock \emph{Quantum theory of the diplon},
\newblock Proceedings of the Royal Society \textbf{A 148}, 146 (1935).

\bibitem{10.2307/24714232}
S.~Albeverio, F.~Gesztesy, R.~Hoegh-Krohn and W.~Kirsch,
\newblock \emph{On point interactions in one dimension},
\newblock Journal of Operator Theory \textbf{12}(1), 101 (1984).

\bibitem{berezin1960remark}
L.~D.~F. F.~A.~Berezin,
\newblock \emph{Remark on the schr\"odinger equation with singular potential},
\newblock Dokl. Akad. Nauk SSSR \textbf{137}, 1011 (1961).

\bibitem{albeverio1979singular}
S.~Albeverio, J.~Fenstad and R.~Hoegh-Krohn,
\newblock \emph{Singular perturbations and nonstandard analysis},
\newblock Transactions of The American Mathematical Society - TRANS AMER MATH
  SOC \textbf{252} (1979),
\newblock \doi{10.2307/1998089}.

\bibitem{albeverio1981point}
R.~H.-K. S.~Albeverio,
\newblock \emph{Point interactions as limits of short range interactions},
\newblock Journal of Operator Theory \textbf{6}, 313 (1981).

\bibitem{cheon1999fermion}
T.~Cheon and T.~Shigehara,
\newblock \emph{Fermion-boson duality of one-dimensional quantum particles with
  generalized contact interactions},
\newblock Phys. Rev. Lett. \textbf{82}, 2536 (1999),
\newblock \doi{10.1103/PhysRevLett.82.2536}.

\bibitem{Yukalov2005FermiBoseMF}
V.~I. Yukalov and M.~D. Girardeau,
\newblock \emph{Fermi-bose mapping for one-dimensional bose gases},
\newblock Laser Physics Letters \textbf{2}, 375 (2005),
\newblock \doi{10.1002/lapl.200510011}.

\bibitem{granet2022duality}
E.~Granet, B.~Bertini and F.~H.~L. Essler,
\newblock \emph{Duality between weak and strong interactions in quantum gases},
\newblock Phys. Rev. Lett. \textbf{128}, 021604 (2022),
\newblock \doi{10.1103/PhysRevLett.128.021604}.

\bibitem{Sen2003TheFL}
D.~Sen,
\newblock \emph{The fermionic limit of the $\delta$-function bose gas: a
  pseudopotential approach},
\newblock Journal of Physics A \textbf{36}, 7517 (2003),
\newblock \doi{10.1088/0305-4470/36/27/305}.

\bibitem{SEBA1986111}
P.~Seba,
\newblock \emph{Some remarks on the $\delta'$-interaction in one dimension},
\newblock Reports on Mathematical Physics \textbf{24}(1), 111 (1986),
\newblock \doi{https://doi.org/10.1016/0034-4877(86)90045-5}.

\bibitem{kurasov1998finite}
P.~Kurasov and J.~Boman,
\newblock \emph{Finite rank singular perturbations and distributions with
  discontinuous test functions},
\newblock Proceedings of the American Mathematical Society \textbf{126} (1998),
\newblock \doi{10.1090/S0002-9939-98-04291-9}.

\bibitem{Seba1986TheGP}
P.~Seba,
\newblock \emph{The generalized point interaction in one dimension},
\newblock Czechoslovak Journal of Physics B \textbf{36}, 667 (1986).

\bibitem{kurasov1996distribution}
P.~Kurasov,
\newblock \emph{Distribution theory for discontinuous test functions and
  differential operators with generalized coefficients},
\newblock Journal of Mathematical Analysis and Applications \textbf{201}(1),
  287 (1996),
\newblock \doi{10.1006/jmaa.1996.0256}.

\bibitem{albeverio1998symmetries}
S.~Albeverio, L.~Dabrowski and P.~Kurasov,
\newblock \emph{Symmetries of schr\"odinger operator with point interactions},
\newblock Letters in Mathematical Physics \textbf{45}, 33 (1998),
\newblock \doi{10.1023/A:1007493325970}.

\bibitem{Coutinho1997GeneralizedPI}
F.~A.~B. Coutinho, Y.~Nogami and J.~F. Perez,
\newblock \emph{Generalized point interactions in one-dimensional quantum
  mechanics},
\newblock Journal of Physics A \textbf{30}, 3937 (1997).

\bibitem{carreau1999four}
M.~Carreau,
\newblock \emph{Four-parameter point-interaction in 1d quantum systems},
\newblock Journal of Physics A: Mathematical and General \textbf{26}, 427
  (1999),
\newblock \doi{10.1088/0305-4470/26/2/025}.

\bibitem{Romn1996TheRF}
J.~M. Rom{\'a}n and R.~Tarrach,
\newblock \emph{The regulated four-parameter one-dimensional point
  interaction},
\newblock Journal of Physics A \textbf{29}, 6073 (1996),
\newblock \doi{10.1088/0305-4470/29/18/033}.

\bibitem{albeverio2000approximation}
L.~N. S.~Albeverio,
\newblock \emph{Approximation of general zero-range potentials},
\newblock Ukrainian Mathematical Journal \textbf{52}, 664 (2000).

\bibitem{granet2022regularization}
E.~Granet,
\newblock \emph{Regularization of a strong-weak duality between pointlike
  interactions in one dimension},
\newblock arxiv:2112.14684  (2021).

\bibitem{gitman2012self}
I.~T. D.M.~Gitman and B.~Voronov,
\newblock \emph{Self-adjoint extensions in quantum mechanics},
\newblock Springer (2012).

\bibitem{albeverio1988solvable}
R.~H.-K. H.~H. S.~Albeverio, F.~Gesztesy,
\newblock \emph{Solvable models in quantum mechanics},
\newblock Springer,
\newblock \doi{10.1017/CBO9780511628832} (1988).

\bibitem{Cheon:1997rx}
T.~Cheon and T.~Shigehara,
\newblock \emph{{Realizing discontinuous wave function with renormalized short
  range potentials}},
\newblock Phys. Lett. A \textbf{243}, 111 (1998),
\newblock \doi{10.1016/S0375-9601(98)00188-1}.

\bibitem{SCHOLTZ199274}
F.~Scholtz, H.~Geyer and F.~Hahne,
\newblock \emph{Quasi-hermitian operators in quantum mechanics and the
  variational principle},
\newblock Annals of Physics \textbf{213}(1), 74 (1992),
\newblock \doi{https://doi.org/10.1016/0003-4916(92)90284-S}.

\bibitem{Mostafazadeh:2001nr}
A.~Mostafazadeh,
\newblock \emph{{PseudoHermiticity versus PT symmetry 2. A Complete
  characterization of nonHermitian Hamiltonians with a real spectrum}},
\newblock J. Math. Phys. \textbf{43}, 2814 (2002),
\newblock \doi{10.1063/1.1461427},
\newblock \eprint{math-ph/0110016}.

\bibitem{mostafazadeh2010pseudo}
A.~Mostafazadeh,
\newblock \emph{Pseudo-hermitian representation of quantum mechanics},
\newblock International Journal of Geometric Methods in Modern Physics
  \textbf{7}(07), 1191 (2010),
\newblock \doi{10.1142/S0219887810004816}.

\bibitem{ashida2020non}
Y.~Ashida, Z.~Gong and M.~Ueda,
\newblock \emph{Non-hermitian physics},
\newblock Advances in Physics \textbf{69}(3), 249 (2020),
\newblock \doi{10.1080/00018732.2021.1876991}.

\bibitem{bender1998real}
C.~M. Bender and S.~Boettcher,
\newblock \emph{Real spectra in non-hermitian hamiltonians having
  $\mathcal{P}\mathcal{T}$ symmetry},
\newblock Phys. Rev. Lett. \textbf{80}, 5243 (1998),
\newblock \doi{10.1103/PhysRevLett.80.5243}.

\bibitem{bender2007making}
C.~Bender,
\newblock \emph{Making sense of non-hermitian hamiltonians},
\newblock Reports on Progress in Physics \textbf{70} (2007),
\newblock \doi{10.1088/0034-4885/70/6/R03}.

\bibitem{Mostafazadeh:2003gz}
A.~Mostafazadeh,
\newblock \emph{{Exact PT symmetry is equivalent to Hermiticity}},
\newblock J. Phys. A \textbf{36}, 7081 (2003),
\newblock \doi{10.1088/0305-4470/36/25/312},
\newblock \eprint{quant-ph/0304080}.

\bibitem{rotter2009topical}
I.~Rotter,
\newblock \emph{A non-hermitian hamilton operator and the physics of open
  quantum systems},
\newblock Journal of Physics A-mathematical and Theoretical - J PHYS A-MATH
  THEOR \textbf{42} (2009),
\newblock \doi{10.1088/1751-8113/42/15/153001}.

\bibitem{cohen1968}
C.~Cohen-Tannoudji,
\newblock \emph{Cargese Lectures in Physics},
\newblock Gordon and Breach, New York (1968).

\bibitem{RIZOV198559}
V.~Rizov, H.~Sazdjian and I.~Todorov,
\newblock \emph{On the relativistic quantum mechanics of two interacting
  spinless particles},
\newblock Annals of Physics \textbf{165}(1), 59 (1985),
\newblock \doi{https://doi.org/10.1016/S0003-4916(85)80005-1}.

\bibitem{Griffiths_1993}
D.~J. Griffiths,
\newblock \emph{Boundary conditions at the derivative of a delta function},
\newblock Journal of Physics A: Mathematical and General \textbf{26}(9), 2265
  (1993),
\newblock \doi{10.1088/0305-4470/26/9/021}.

\bibitem{Coutinho_2005}
F.~A.~B. Coutinho, Y.~Nogami, L.~Tomio and F.~M. Toyama,
\newblock \emph{Energy-dependent point interactions in one dimension},
\newblock Journal of Physics A: Mathematical and General \textbf{38}(22), 4989
  (2005),
\newblock \doi{10.1088/0305-4470/38/22/020}.

\bibitem{coutinho2006energy}
L.~T.-F.~T. F.A.B.~Coutinho, Y.~Nogami,
\newblock \emph{Energy-dependent point interaction: self-adjointness},
\newblock Canadian Journal of Physics \textbf{84} (2006),
\newblock \doi{10.1063/1.4936302}.

\bibitem{Lange2014DistributionTF}
R.-J. Lange,
\newblock \emph{Distribution theory for schr\"odinger's integral equation},
\newblock Journal of Mathematical Physics \textbf{56} (2014),
\newblock \doi{10.1063/1.4936302}.

\bibitem{girardeau1960relationship}
M.~Girardeau,
\newblock \emph{Relationship between systems of impenetrable bosons and
  fermions in one dimension},
\newblock Journal of Mathematical Physics \textbf{1}, 516 (1960).

\bibitem{bouchoule2020the}
I.~Bouchoule, B.~Doyon and J.~Dubail,
\newblock \emph{The effect of atom losses on the distribution of rapidities in
  the one-dimensional bose gas},
\newblock SciPost Physics \textbf{9} (2020),
\newblock \doi{10.21468/SciPostPhys.9.4.044}.

\end{thebibliography}

%
%
%
%
%

\end{document}